\documentclass[submission,copyright,creativecommons,sharealike,noncommercial]{eptcs}
\providecommand{\event}{ACT 2020}
\pdfoutput=1

\usepackage[utf8]{inputenc}
\usepackage{amsmath,amsthm,amssymb,amsfonts}
\usepackage{graphicx}
\usepackage{wrapfig}
\usepackage{url}
\usepackage{mathtools, nccmath}
\usepackage[font={sf}]{caption}
\usepackage{tikz}
\usepackage{tikzit}
\usepackage{stmaryrd}
\usepackage{physics}

\tikzstyle{doubled}=[line width=1.5pt] 

\tikzstyle{dot}=[inner sep=0mm,minimum width=2mm,minimum height=2mm,draw,shape=circle]  
\tikzstyle{ddot}=[inner sep=0mm, doubled, minimum width=2.5mm,minimum height=2.5mm,draw,shape=circle]

\tikzstyle{pdot}=[inner sep=0mm, doubled, minimum width=2.5mm,minimum height=2.5mm,shape=circle]
\tikzstyle{phase dimensions}=[minimum size=6mm,font=\footnotesize,inner sep=0.2mm,outer sep=-2mm]

\tikzstyle{phase dot}=[pdot,phase dimensions]
\tikzstyle{wphase dot}=[dot, phase dimensions]

\tikzstyle{hadamard}=[fill=white,draw,inner sep=0.6mm,font=\footnotesize,minimum height=6mm,minimum width=8mm]

\tikzstyle{anti} = [fill=white,draw,inner sep=0.6mm,font=\footnotesize,minimum height=3mm,minimum width=3mm]

\tikzstyle{triang}=[regular polygon,regular polygon sides=3,draw,scale=0.75,inner sep=-0.75pt,minimum width=9mm,fill=white,regular polygon rotate=180]
\tikzstyle{triang_lesssep}=[regular polygon,regular polygon sides=3,draw,scale=0.75,inner sep=-4pt,minimum width=9mm,fill=white,regular polygon rotate=180, text depth=4mm]
\tikzstyle{triangdag}=[regular polygon,regular polygon sides=3,draw,scale=0.75,inner sep=-0.5pt,minimum width=9mm,fill=white]

\makeatletter
\newcommand{\boxshape}[3]{%
\pgfdeclareshape{#1}{
\inheritsavedanchors[from=rectangle] 
\inheritanchorborder[from=rectangle]
\inheritanchor[from=rectangle]{center}
\inheritanchor[from=rectangle]{north}
\inheritanchor[from=rectangle]{south}
\inheritanchor[from=rectangle]{west}
\inheritanchor[from=rectangle]{east}
\backgroundpath{
\southwest \pgf@xa=\pgf@x \pgf@ya=\pgf@y
\northeast \pgf@xb=\pgf@x \pgf@yb=\pgf@y

\@tempdima=#2
\@tempdimb=#3

\pgfpathmoveto{\pgfpoint{\pgf@xa - 5pt + \@tempdima}{\pgf@ya}}
\pgfpathlineto{\pgfpoint{\pgf@xa - 5pt - \@tempdima}{\pgf@yb}}
\pgfpathlineto{\pgfpoint{\pgf@xb + 5pt + \@tempdimb}{\pgf@yb}}
\pgfpathlineto{\pgfpoint{\pgf@xb + 5pt - \@tempdimb}{\pgf@ya}}
\pgfpathlineto{\pgfpoint{\pgf@xa - 5pt + \@tempdima}{\pgf@ya}}
\pgfpathclose
}
}}

\boxshape{NEbox}{0pt}{5pt}
\boxshape{SEbox}{0pt}{-5pt}
\boxshape{NWbox}{5pt}{0pt}
\boxshape{SWbox}{-5pt}{0pt}
\boxshape{EBox}{-3pt}{3pt}
\boxshape{WBox}{3pt}{-3pt}
\makeatother

\tikzstyle{map}=[draw,shape=NEbox,inner sep=2pt,minimum height=6mm,fill=white]
\tikzstyle{mapdag}=[draw,shape=SEbox,inner sep=2pt,minimum height=6mm,fill=white]
\tikzstyle{maptrans}=[draw,shape=SWbox,inner sep=2pt,minimum height=6mm,fill=white]
\tikzstyle{mapconj}=[draw,shape=NWbox,inner sep=2pt,minimum height=6mm,fill=white]

\tikzstyle{dmap}=[draw,doubled,shape=NEbox,inner sep=2pt,minimum height=6mm,fill=white]
\tikzstyle{dmapdag}=[draw,doubled,shape=SEbox,inner sep=2pt,minimum height=6mm,fill=white]
\tikzstyle{dmaptrans}=[draw,doubled,shape=SWbox,inner sep=2pt,minimum height=6mm,fill=white]
\tikzstyle{dmapconj}=[draw,doubled,shape=NWbox,inner sep=2pt,minimum height=6mm,fill=white]

\makeatletter
\pgfdeclareshape{cornerpoint}{
\inheritsavedanchors[from=rectangle] 
\inheritanchorborder[from=rectangle]
\inheritanchor[from=rectangle]{center}
\inheritanchor[from=rectangle]{north}
\inheritanchor[from=rectangle]{south}
\inheritanchor[from=rectangle]{west}
\inheritanchor[from=rectangle]{east}
\backgroundpath{
\southwest \pgf@xa=\pgf@x \pgf@ya=\pgf@y
\northeast \pgf@xb=\pgf@x \pgf@yb=\pgf@y

\pgfmathsetmacro{\pgf@shorten@left}{\pgfkeysvalueof{/tikz/shorten left}}
\pgfmathsetmacro{\pgf@shorten@right}{\pgfkeysvalueof{/tikz/shorten right}}

\pgfpathmoveto{\pgfpoint{0.5 * (\pgf@xa + \pgf@xb)}{\pgf@ya - 5pt}}
\pgfpathlineto{\pgfpoint{\pgf@xa - 8pt + \pgf@shorten@left}{\pgf@yb - 1.5 * \pgf@shorten@left}}
\pgfpathlineto{\pgfpoint{\pgf@xa - 8pt + \pgf@shorten@left}{\pgf@yb}}
\pgfpathlineto{\pgfpoint{\pgf@xb + 8pt - \pgf@shorten@right}{\pgf@yb}}
\pgfpathlineto{\pgfpoint{\pgf@xb + 8pt - \pgf@shorten@right}{\pgf@yb - 1.5 * \pgf@shorten@right}}
\pgfpathclose
}
}

\pgfdeclareshape{cornercopoint}{
\inheritsavedanchors[from=rectangle] 
\inheritanchorborder[from=rectangle]
\inheritanchor[from=rectangle]{center}
\inheritanchor[from=rectangle]{north}
\inheritanchor[from=rectangle]{south}
\inheritanchor[from=rectangle]{west}
\inheritanchor[from=rectangle]{east}
\backgroundpath{
\southwest \pgf@xa=\pgf@x \pgf@ya=\pgf@y
\northeast \pgf@xb=\pgf@x \pgf@yb=\pgf@y

\pgfmathsetmacro{\pgf@shorten@left}{\pgfkeysvalueof{/tikz/shorten left}}
\pgfmathsetmacro{\pgf@shorten@right}{\pgfkeysvalueof{/tikz/shorten right}}

\pgfpathmoveto{\pgfpoint{0.5 * (\pgf@xa + \pgf@xb)}{\pgf@yb + 5pt}}
\pgfpathlineto{\pgfpoint{\pgf@xa - 8pt + \pgf@shorten@left}{\pgf@ya + 1.5 * \pgf@shorten@left}}
\pgfpathlineto{\pgfpoint{\pgf@xa - 8pt + \pgf@shorten@left}{\pgf@ya}}
\pgfpathlineto{\pgfpoint{\pgf@xb + 8pt - \pgf@shorten@right}{\pgf@ya}}
\pgfpathlineto{\pgfpoint{\pgf@xb + 8pt - \pgf@shorten@right}{\pgf@ya + 1.5 * \pgf@shorten@right}}
\pgfpathclose
}
}

\makeatother

\pgfkeyssetvalue{/tikz/shorten left}{0pt}
\pgfkeyssetvalue{/tikz/shorten right}{0pt}

\tikzstyle{kpoint common}=[draw,fill=white,inner sep=1pt,minimum height=4mm]
\tikzstyle{kpoint}=[shape=cornerpoint,shorten left=5pt,kpoint common]
\tikzstyle{kpoint adjoint}=[shape=cornercopoint,shorten left=5pt,kpoint common]
\tikzstyle{kpoint conjugate}=[shape=cornerpoint,shorten right=5pt,kpoint common]
\tikzstyle{kpoint transpose}=[shape=cornercopoint,shorten right=5pt,kpoint common]

\tikzstyle{kpointdag}=[kpoint adjoint]
\tikzstyle{kpointadj}=[kpoint adjoint]
\tikzstyle{kpointconj}=[kpoint conjugate]
\tikzstyle{kpointtrans}=[kpoint transpose]

\tikzstyle{big kpoint}=[kpoint, minimum width=1.0 cm, minimum height=2mm, inner sep=4pt, text depth=1.5mm]

 \tikzstyle{upground}=[circuit ee IEC,thick,ground,rotate=90,scale=1.5]
 \tikzstyle{downground}=[circuit ee IEC,thick,ground,rotate=-90,scale=1.5]

\tikzstyle{smallcirc}=[circle,fill=white,draw=black]
\tikzstyle{plain}=[-,draw=black,line width=2.000]
\tikzstyle{process}=[rectangle,fill=white,draw=black]
\tikzstyle{none}=[inner sep=0pt]

\tikzstyle{discarding}=[fill=white, draw=black, shape=circle, style=upground]
\tikzstyle{smalldiscarding}=[fill=white, draw=black, style=upground, scale=0.5]
\tikzstyle{backdiscard}=[fill=white, draw=black, shape=circle, style=downground, scale=0.5]
\tikzstyle{smallbackdiscard}=[fill=white, draw=black, shape=circle, style=downground, scale=0.5]
\tikzstyle{state}=[fill=white, draw=black, style=triang, tikzit shape=rectangle]
\tikzstyle{kstate}=[fill=white, draw=black, style=kpoint, tikzit shape=rectangle]
\tikzstyle{kstateconj}=[fill=white, draw=black, style=kpoint conjugate, tikzit shape=rectangle]
\tikzstyle{kstateBIG}=[fill=white, draw=black, style=big kpoint, tikzit shape=rectangle]
\tikzstyle{effect}=[fill=white, draw=black, style=triangdag]
\tikzstyle{keffect}=[fill=white, draw=black, style=kpoint adjoint]
\tikzstyle{keffectconj}=[fill=white, draw=black, style=kpoint transpose]
\tikzstyle{morphdag}=[style=mapdag]
\tikzstyle{morph}=[style=hadamard]
\tikzstyle{WIDEmorph}=[style=hadamard, minimum width=14mm]
\tikzstyle{morphtrans}=[style=maptrans]
\tikzstyle{morphconj}=[style=mapconj]
\tikzstyle{CPMmorph}=[style=dmap]
\tikzstyle{CPMmorphconj}=[style=dmapconj]
\tikzstyle{CPMmorphdag}=[style=dmapdag]
\tikzstyle{CPMmorphtrans}=[style=dmaptrans]
\tikzstyle{CPMstate}=[fill=white, draw=black, style=triang, doubled]
\tikzstyle{CPMstateBIG}=[fill=white, draw=black, style={triang_lesssep}, doubled]
\tikzstyle{CPMkstate}=[fill=white, draw=black, style=kpoint, tikzit shape=rectangle, doubled]
\tikzstyle{CPMkstateconj}=[fill=white, draw=black, style=kpoint conjugate, tikzit shape=rectangle, doubled]
\tikzstyle{CPMkstateBIG}=[fill=white, draw=black, style=big kpoint, tikzit shape=rectangle, doubled]
\tikzstyle{CPMkeffect}=[fill=white, draw=black, style=kpoint adjoint, doubled]
\tikzstyle{CPMkeffectconj}=[fill=white, draw=black, style=kpoint transpose, doubled]
\tikzstyle{UHfB}=[fill=white, draw=black, style=triangdag, doubled, inner sep=-2pt]
\tikzstyle{leak}=[style=tinypoint, regular polygon rotate=-90]
\tikzstyle{leakfill}=[style=tinypoint, regular polygon rotate=-90, fill=black]
\tikzstyle{Z}=[style=dot, fill=green]
\tikzstyle{X}=[style=dot, fill=red]
\tikzstyle{black_dot}=[style=dot, fill=black]
\tikzstyle{white_dot}=[style=dot, fill=white]
\tikzstyle{qblack_dot}=[style=ddot, fill=black]
\tikzstyle{qwhite_dot}=[style=ddot, fill=white]
\tikzstyle{whitephase}=[style=wphase dot, fill=white]
\tikzstyle{qredphase}=[style=phase dot, fill=red]
\tikzstyle{qgreenphase}=[style=phase dot, fill=green]
\tikzstyle{had}=[style=hadamard, doubled]
\tikzstyle{box}=[style=hadamard]
\tikzstyle{classhad}=[style=hadamard]
\tikzstyle{antipode}=[style=anti]
\tikzstyle{greydot}=[fill={rgb,255: red,128; green,128; blue,128}, draw=black, shape=circle]

\tikzstyle{dottededge}=[-, dotted]
\tikzstyle{double edge}=[-, style=doubled, draw=black, tikzit draw={rgb,255: red,191; green,0; blue,64}]
\tikzstyle{new edge style 0}=[<-]

\usetikzlibrary{calc, positioning, shapes.geometric}
\usetikzlibrary{
	arrows,
	shapes,
	decorations,
	intersections,
	backgrounds,
	positioning,
	circuits.ee.IEC
	}

\newcommand{\tikzfigscale}[2]{\scalebox{#1}{\input{#2.tikz}}}

\newcommand{\cat}{\mathbf}
\newcommand{\morph}[1]{\xrightarrow{#1}}
\newcommand{\id}[1]{\textrm{id}_{#1}}
\newcommand{\ob}[1]{\textrm{Ob}(#1)}
\newcommand{\fhilb}{\textbf{FHilb}}
\newcommand{\frel}{\textbf{FRel}}
\newcommand{\fset}{\textbf{FSet}}
\newcommand{\cpm}[1]{\mathbf{CPM}(#1)}
\newcommand{\spl}[1]{\mathbf{Split}(#1)}
\newcommand{\set}{\textbf{Set}}

\newcommand{\putt}[1]{\mathbin{\uparrow_{#1}}}
\newcommand{\get}{\raisebox{-0.01cm}{\includegraphics[scale=0.02]{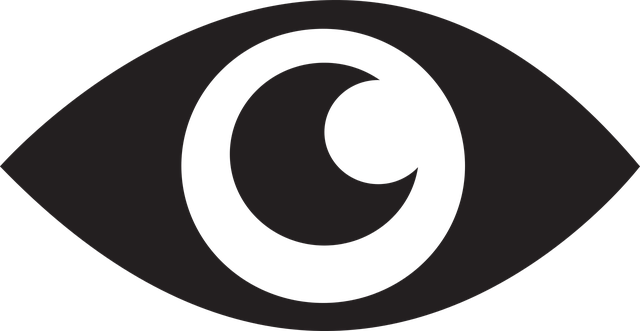}}}
\newcommand{\mix}[1]{\sim_{#1}}
\newcommand{\copyy}[1]{
\mathbin{\begin{tikzpicture}
		\node [style={black_dot}, scale=0.6] (0) at (0, 0) {};
		\node [style=none] (1) at (0, -0.2) {};
		\node [style=none] (2) at (-0.2, 0.2) {};
		\node [style=none] (3) at (0.2, 0.2) {};
		\draw (1.center) to (0);
		\draw [in=-90, out=150] (0) to (2.center);
		\draw [in=270, out=30] (0) to (3.center);
\end{tikzpicture}}_{\hspace{-0.05cm}#1}}

\newcommand{\wcopyy}{
\mathbin{\begin{tikzpicture}
		\node [style={white_dot}, scale=0.6] (0) at (0, 0) {};
		\node [style=none] (1) at (0, -0.2) {};
		\node [style=none] (2) at (-0.2, 0.2) {};
		\node [style=none] (3) at (0.2, 0.2) {};
		\draw (1.center) to (0);
		\draw [in=-90, out=150] (0) to (2.center);
		\draw [in=270, out=30] (0) to (3.center);
\end{tikzpicture}}}

\newcommand{\unit}{
\mathbin{\begin{tikzpicture}
	\begin{pgfonlayer}{nodelayer}
		\node [style={black_dot}, scale=0.6] (0) at (0, -0.1) {};
		\node [style=none] (1) at (0, 0.2) {};
	\end{pgfonlayer}
	\begin{pgfonlayer}{edgelayer}
		\draw (1.center) to (0);
	\end{pgfonlayer}
\end{tikzpicture}}}

\newcommand{\counit}[1]{
\mathbin{
\begin{tikzpicture}
	\begin{pgfonlayer}{nodelayer}
		\node [style={black_dot}, scale=0.6] (0) at (0, 0.1) {};
		\node [style=none] (1) at (0, -0.2) {};
	\end{pgfonlayer}
	\begin{pgfonlayer}{edgelayer}
		\draw (1.center) to (0);
	\end{pgfonlayer}
\end{tikzpicture}}_{\hspace{-0.05cm}#1}}

\newcommand{\effect}[1]{
\mathbin{\begin{tikzpicture}
	\begin{pgfonlayer}{nodelayer}
		\node [style=effect, scale=0.6] (0) at (0, 0.125) {};
		\node [style=none] (1) at (0, -0.15) {};
	\end{pgfonlayer}
	\begin{pgfonlayer}{edgelayer}
		\draw (1.center) to (0);
	\end{pgfonlayer}
\end{tikzpicture}}_{\hspace{-0.05cm}#1}}

\newcommand{\deco}{
\mathbin{\begin{tikzpicture}
	\begin{pgfonlayer}{nodelayer}
		\node [style=none] (1) at (-0.3, 0.2) {};
		\node [style=none] (2) at (-0.3, -0.2) {};
		\node [style=none] (3) at (0.2, -0.2) {};
		\node [style={black_dot}, scale=0.6] (4) at (-0.3, 0) {};
		\node [style=none] (5) at (0.2, 0.2) {};
		\node [style={black_dot}, scale=0.6] (6) at (0.2, 0) {};
	\end{pgfonlayer}
	\begin{pgfonlayer}{edgelayer}
		\draw (2.center) to (1.center);
		\draw (5.center) to (3.center);
		\draw (4) to (6);
	\end{pgfonlayer}
\end{tikzpicture}}}

\newcommand{\frob}{
\mathbin{\begin{tikzpicture}
	\begin{pgfonlayer}{nodelayer}
		\node [style={black_dot}, scale=0.6] (0) at (0, 0) {};
		\node [style=none] (1) at (0, 0.2) {};
		\node [style=none] (2) at (-0.2, -0.2) {};
		\node [style=none] (3) at (0.2, -0.2) {};
	\end{pgfonlayer}
	\begin{pgfonlayer}{edgelayer}
		\draw (1.center) to (0);
		\draw [in=90, out=-150] (0) to (2.center);
		\draw [in=-270, out=-30] (0) to (3.center);
	\end{pgfonlayer}
\end{tikzpicture}}}

\newcommand{\wfrob}{
\mathbin{\begin{tikzpicture}
	\begin{pgfonlayer}{nodelayer}
		\node [style={white_dot}, scale=0.6] (0) at (0, 0) {};
		\node [style=none] (1) at (0, 0.2) {};
		\node [style=none] (2) at (-0.2, -0.2) {};
		\node [style=none] (3) at (0.2, -0.2) {};
	\end{pgfonlayer}
	\begin{pgfonlayer}{edgelayer}
		\draw (1.center) to (0);
		\draw [in=90, out=-150] (0) to (2.center);
		\draw [in=-270, out=-30] (0) to (3.center);
	\end{pgfonlayer}
\end{tikzpicture}}}

\theoremstyle{definition}
\newtheorem{defn}{Definition}
\theoremstyle{plain}
\newtheorem{prop}{Proposition}
\theoremstyle{plain}
\newtheorem{corr}{Corollary}

\begin{document}
\title{The Safari of Update Structures: Visiting the Lens and Quantum Enclosures}

\author{Matthew Wilson \institute{University of Oxford} \institute{HKU-Oxford Joint Laboratory for \\ Quantum Information and Computation \email{matthew.wilson@cs.ox.ac.uk}} \and
James Hefford \institute{University of Oxford} \email{james.hefford@cs.ox.ac.uk} \and Guillaume Boisseau \institute{University of Oxford} \email{guillaume.boisseau@cs.ox.ac.uk} \and 
Vincent Wang  \institute{University of Oxford} \email{vincent.wang@cs.ox.ac.uk}}

\def\titlerunning{The Safari of Update Structures}
\def\authorrunning{M. Wilson, J. Hefford, G. Boisseau and V. Wang}

\maketitle

\begin{abstract}
    We build upon our recently introduced concept of an update structure to show that it is a generalisation of very-well-behaved lenses, that is, there is a bijection between a strict subset of update structures and vwb lenses in cartesian categories. We show that update structures are also sufficiently general to capture quantum observables, pinpointing the additional assumptions required to make the two coincide. In doing so, we shift the focus from special commutative $\dag$-Frobenius algebras to interacting (co)magma (co)module pairs, showing that the algebraic properties of the (co)multiplication arise from the module-comodule interaction, rather than direct assumptions about the magma-comagma pair. We then begin to investigate the zoo of possible update structures, introducing the notions of classical security-flagged databases, and databases of quantum systems. This work is of foundational interest as update structures place previously distinct areas of research in a general class of operationally motivated structures, we expect the taming of this class to illuminate novel relationships between separately studied topics in computer science, physics and mathematics.
\end{abstract}

\section{Introduction}
Modelling meaning updating within natural language processing is an ongoing foundational problem. Inspired by an operational interpretation, we recently introduced the concept of an \textit{update structure} \cite{hefford2020categories} as a candidate for modelling meaning updating in monoidal categories, with the intention of possible applications within DisCoCat \cite{coecke_mathematics_2019, coecke_mathematical_meaning,coecke2020meaning}. At its core, an update structure is a magma-module and comagma-comodule, with axioms adopted to govern their interactions in accordance with some desired operational intuition.\\

Particular cases of module-comodule pairs have been considered elsewhere; they make an appearance in both categorical quantum mechanics where they have been shown to capture quantum observables and measurements \cite{coecke_measurements, heunen2019categories} and in work on lenses \cite{bancilhon_update_1981,foster_combinators_nodate} in relation to the view-update problem. Additionally, the connections between monads and modules (algebras over a monad) are well-known with studies into monads in dagger categories \cite{heunen_monads} allowing for applications in the monadic dynamics framework \cite{gogioso_monadic} where modules can describe the evolution of quantum systems under the Schr\"odinger equation \cite{gogioso_schrodinger,gogioso_dynamics}.\\

In this paper we offer a unifying perspective of these fields from the vantage point of update structures. To this end, the purpose of the remainder of this article is twofold: to pinpoint precisely when lenses and update structures coincide, and to demonstrate that quantum measurements can be used to define update structures. For the former, we demonstrate that update structures are more general than vwb lenses, capturing both demolition and non-demolition viewing (Get) processes. For the latter, we demonstrate both that the application of decoherence to any update structure will produce a new update structure on classical objects, and that to characterise quantum observables it is sufficient to look at how the weak module-comodule interact without underlying assumptions on the magma-comagma pair\footnote{This is contrary to the usual approach where one typically begins with a module-comodule over a $\dag$-Frobenius structure as a starting assumption.}. Thus the focus is shifted from demands on both how the module and comodule interact \textit{and} the algebraic structure of the (co)magma, to demands placed purely on the interaction between module and comodule. This approach places quantum observables and vwb lenses within a zoo of update structures with many other enclosures left to explore, examples of which we give in the last section.

\subsection{Introduction to Update Structures}

We suppose that in a monoidal category, we may encode `states' of objects via morphisms into that object from the tensor unit. Given a pair of objects designated `system' and `property (of the system)', an update structure is a tuple of morphisms $(\putt{},\get{},\mix{},\copyy{})$ which come with four axioms chosen to force the morphisms to behave as if they are update, read-out, pre-processing, and copying procedures respectively on the system and property. We borrow the language of categorical lenses throughout, as it reflects well on the modelling intentions of these operations. We recall some definitions from \cite{hefford2020categories}.

\begin{defn}[(Strong) Update Structure]
An update structure $(\putt{},\get{},\mix{},\copyy{})$ in a monoidal category $\cat{C}$ consists of:
\begin{itemize}
    \item An object $S$, which we refer to as a \textbf{system}
    \item An object $p$, which we refer to as a \textbf{property}, which has:
    \begin{itemize}
    \item A magma structure $\mix{}: p \otimes p \rightarrow p$
    \item A comagma structure $\copyy{}: p \rightarrow p \otimes p$
    \end{itemize}
    \item A \textbf{Put} operation $\putt{}: S \otimes p \rightarrow S$
    \item A \textbf{Get} operation $\get{}: S \rightarrow S \otimes p$
\end{itemize}
Which satisfy the following equations:

\begin{equation}\label{putputgetget}
\arraycolsep=1.5cm
\begin{array}{cc}
    \textrm{PutPut} & \textrm{GetGet} \\
    \tikzfigscale{1}{figs/putput} & \tikzfigscale{1}{figs/getget}
\end{array}
\end{equation}
making $\putt{}$ a magma module and $\get{}$ a comagma module. Additionally we require:

\begin{equation}\label{putgetgetput}
\arraycolsep=1.5cm
\begin{array}{cc}
    \textrm{PutGet} & \textrm{GetPut} \\
    \tikzfigscale{1}{figs/putget} &\tikzfigscale{1}{figs/getput}
\end{array}
\end{equation}
\end{defn}

These four equations are operationally well-chosen with the intent of capturing what it means to be an ``update''. The PutPut rule says updating twice ought to be the same as performing some operation on the properties and then updating once with the new combined property. GetGet is the converse of this; retrieving a property twice should be the same as retrieving once and then performing some operation (e.g. copying) on the retrieved data. PutGet captures the notion that updating and then retrieving ought to be the same as copying, then updating with one of the copies. Finally GetPut is the notion that looking and then putting back ought to leave the system invariant.\\

The first two equations satisfied by update structures are the requirements that $(\putt{},\mix{})$ and $(\get{},\copyy{})$ are a ``weak'' module-comodule pair, that is a pair of an action and coaction on a magma-comagma, without the stricter underlying assumption of a monoid-comonoid pair\footnote{that is, associativity and a unit, and often speciality and the Frobenius law, etc.}. The corresponding operational intuition is that the magma handles the potential modification that may occur to the property stored in the system when the system is updated in succession, and respectively the comagma handles modifications to the property from successive read-outs. For instance, in modelling a faulty memory system, successive property updates may not fully overwrite the previous property of the system, and successive read-outs may corrupt the currently stored property.\\

The remaining two axioms are operational notions which constrain how the update and read-out procedures should interact. Our initial intention was to model processes that faithfully recall stored properties without requiring the update process to behave as an overwrite. So, taking the comagma to be the copy-process on the property, the detailed prose of the PutGet axiom is: ``Updating with a property and then reading-out the system is observationally equivalent to copying the property, updating the system with one copy and observing the other." The GetPut axiom reads: ``Reading-out a system and updating the system with the reading is trivial.", or, ``Updating with the same property twice is trivial."\\

In some circumstances, GetPut is quite a strong demand. For instance, classical agents who extract and re-insert information from a quantum system will disturb the system via decoherence and the GetPut axiom will not hold. Thus we also defined a strictly weaker notion.
\begin{defn}[Weak Update Structure]
A weak update structure $(\putt{},\get{},\mix{},\copyy{})$ satisfies all the axioms of an update structure apart from GetPut which we replace with the strictly weaker \textit{repeat-update} axiom:
\begin{equation}\label{repeatupdate}
\tikzfigscale{1}{figs/repeatupdate}
\end{equation}
\end{defn}
\begin{prop}\cite{hefford2020categories}
Any strong update structure is a weak update structure.
\end{prop}

For a weak update, it is not assumed that extraction and re-insertion of data leaves the system invariant. The Put is however ``static'': updating twice with copies of a property is the same as updating once with that property. This staticness is enough to ensure that disturbing a system twice via reading out and reinserting is the same as just disturbing once, formally $\putt{} \circ \get{}$ is idempotent. In \cite{hefford2020categories} the idempotence of $\putt{} \circ \get{}$ was used to demonstrate that any weak update structure in a category $\cat{C}$ can be used to construct a strong update structure in $\spl{\cat{C}}$, the Karoubi envelope of $\cat{C}$.
We give a separate operational axiom to capture the existence of an update which does nothing to the system.
\begin{defn}[Trivial Update] An update structure has a trivial update if
\begin{equation}\label{prop:discard}
    \tikzfigscale{1}{figs/discardable}
\end{equation}
It has a trivial outcome if $\get{}$ has the same property but flipped vertically.
\end{defn}
Since we introduce the weakening of strong update structures to cope with interactions which disturb or collapse systems in the way that for example quantum measurements do, it is appropriate that those weak updates which have trivial updates are strong.
\begin{prop}\label{prop:weaktostrong}
    Any weak update structure with a trivial update is a strong update structure:
\end{prop}
\begin{proof}
By insertion of the trivial update, PutGet, the repeat-update axiom, and another use of the trivial update.
\begin{equation}\label{ignortogetput}
    \tikzfigscale{1}{figs/ignortogetput}
\end{equation}
\end{proof}
However, it is not the case that all strong update structures have trivial updates. It will be shown in the next section that whilst vwb lenses define strong update structures, any vwb lens with a trivial update has a separable Put morphism, $\putt{} = \id{} \times \counit{} : S \times p \rightarrow S$.\\

It is worth pointing out that the PutPut and GetGet laws place quite strong constraints on the \\ (co)associativity of the (co)magma.  One can deduce the following from PutPut and GetGet:
\begin{equation*}
    \tikzfigscale{1}{figs/putgetcoass}
\end{equation*}
While the definition of an update structure allows for a (co)magma which is not (co)associative, in this case the Put/Get would have to be noisy. More concretely it cannot be \textit{faithful} (see definition \ref{faithful}).

\section{Relation to Very-Well-Behaved Lenses}
In this section we relate update structures to very-well-behaved lenses. Indeed, to those familiar with lenses, our choice of language ``Get'' and ``Put'' may raise suspicions that the two are related. We now make the connection precise.

\begin{defn}[Very-Well-Behaved Lens]
In a category $\cat{C}$ with finite products, a lens \cite{foster_combinators_nodate} is a tuple $(S,V,g:S\rightarrow V,p:S \times V \rightarrow S)$, where $S$ and $V$ are objects of $\cat{C}$. A lens is further \emph{very-well-behaved} (\emph{vwb}) if it satisfies the following:
\begin{itemize}
    \item (PutPut): $p(p(s,v_1),v_2) = p(s,v_2)$
    \item (PutGet): $g(p(s,v)) = v$
    \item (GetPut): $p(s,g(s)) = s$
\end{itemize}
\end{defn}
\begin{prop}
\label{lensequiv}
vwb lenses $(S,V,g,p)$ in a category $\cat{C}$ with finite products are in bijection with update structures that have trivial outcomes, $\mix{} := \pi_2$, and $\copyy{} := \delta_V$. Under this bijection, the very-well-behavedness laws are equivalent to their update structure counterparts.
\end{prop}

\begin{proof}
Recalling that categories with finite products are cartesian monoidal, we rewrite the lens laws in the following suggestive graphical form.

\begin{equation*}
    \tikzfigscale{0.7}{figs/prop1/lawfullens}
\end{equation*}

The bijection identifies $p$ and $\putt{}$. $\copyy{}$ is $\delta_V$, the copy on $V$, and $\mix{}$ is $\pi_2$. $g$ and $\get{}$ follow the obvious type-matching strategy.

\begin{equation*}
    \tikzfigscale{1}{figs/prop1/struct2lensrules}
\end{equation*}
We verify one injection as follows:

\begin{equation*}
    \tikzfigscale{1}{figs/prop1/lensembed}
\end{equation*}

And the other (making use of the fact that $\get{}$ has trivial outcomes in the penultimate step):

\begin{equation*}
    \tikzfigscale{1}{figs/prop1/getinject}
\end{equation*}

From the substitutions on $\mix{}$ and $\get{}$, it immediately follows that the PutPut and GetPut laws are equivalent for both update structures and vwb lenses. Regarding PutGet, we have that PutGet for lenses implies PutGet for the corresponding update structure, and PutGet for an update structure implies PutGet for the corresponding lens:

\begin{equation*}
    \tikzfigscale{1}{figs/prop1/putgetlensequivalence}
\end{equation*}

It is worth noting that this correspondence makes use of the `left-delete' magma $\mix{} = \pi_2$ which is \emph{not} a monoid.
\end{proof}

Thus update structures generalise lenses; every vwb lens is an update structure but not the converse. Perhaps the most important distinction is that update structures allow for flexibility in the way that Get affects the system. When a lens is viewed as an update structure, its Get operation does not modify the system $S$; it duplicates it using the duplication map that comes with the finite product structure, and returns it unmodified alongside the computed $V$. 

In the more general setting of a monoidal category one is not guaranteed the additional structure of a universal copying map, for instance we can at best copy one basis of a quantum system. As such it is natural to require that the Get also returns a system of type $S$, rather than choosing to copy an entire system in a particular basis. Thus Proposition \ref{lensequiv} serves as a bridge between vwb lenses and update structures in monoidal settings: where $S$ and $V$ in the monoidal category are comonoid objects, and $\get{}$ and $\putt{}$ are comonoid homomorphisms, one obtains a close analogue of vwb lenses.

This vantage point of lenses \emph{qua} update structures appears to be an interesting new point in the space of lens-like structures: it generalizes the usual cartesian lenses without being as general as optics \cite{rileyCategoriesOptics2018}. We remark an interesting similarity with Abou-Saleh et al.'s monadic lenses \cite{abou-saleh_reflections_2016}: they also generalize lenses to a setting that is not quite cartesian monoidal, and in doing so come up with strikingly similar laws. Specifically, their MPutGet and MGetPut laws are exact analogues to our PutGet and GetPut, that would be drawn identically if string diagrams could be drawn in the monadic case. They also mention that monadic lenses where the $get$ direction is effectful require an analogue of our GetGet law. The Kleisli category they use is however not monoidal in general, thus their monadic lenses are not quite an instance of update structures.

We note that intuitively vwb lenses represent a case in which the only relevant update is the most recent one, clearly then there should be no interesting vwb lenses with trivial updates. Indeed, any Put in a vwb lens with a trivial update separates as a map which throws the new property away and leaves the system alone.  

\begin{equation*}
    \tikzfigscale{1}{figs/ignorevwb}
\end{equation*}


An additional distinguishing point between lenses and update structures appears relative to composition. Notions of lenses all have the important property of composability: one can compose lenses to build accessors for a big data type in terms of accessors for its subcomponents. In this paper we however do not investigate composition of update structures, and it appears there might be different sensible such notions\footnote{with one possibility appearing in \cite{hefford2020categories}}, requiring different additional structure on the (co)magmas and operations.

\section{Interacting Module-Comodule Pairs}
Given a (co)magma (co)module pair, there are many possible PutGet rules one could impose and in this section we will study a handful of these possibilities, elucidating the restrictions each choice places on the rest of the structure, in particular the magma-comagma pair.

\begin{equation*}
\arraycolsep=0.8cm
\begin{array}{ccc}
    A & B & C \\
    \tikzfigscale{1}{figs/putget1} & \tikzfigscale{1}{figs/putget2} &
    \tikzfigscale{1}{figs/putget3}
\end{array}
\end{equation*}

\begin{prop}\label{prop:putgetA}
    In the presence of PutPut and GetGet, PutGet A implies:
    \begin{enumerate}
    \item with the GetPut rule, $S\simeq S\otimes p$, which in many categories, for instance $\fhilb$, $\frel$ and $\fset$, means $S$ or $p$ have to in some way be trivial. In $\fhilb$ $p=I$, and in $\frel$ and $\fset$, either $p=I$, $p=\varnothing$ or $S=\varnothing$.
    \item If the (co)magma has a (co)unit then the identity separates.

    \end{enumerate}
\end{prop}
\begin{proof}
\begin{enumerate}
\item In the presence of the GetPut law we immediately have that $S\simeq S\otimes p$. In $\fhilb$ this requires $p=I$ and in $\frel$ and $\fset$ either $p=I$, $p=\varnothing$ or $S=\varnothing$. Thus the update is rendered essentially trivial.
\item By looking at $(\get{}\otimes \id{p})\circ\get{}\circ \putt{}$ and $\get{}\circ\putt{}\circ(\putt{}\otimes\id{p})$ one can show
\begin{equation*}
    \tikzfigscale{1}{figs/putgeta1}
\end{equation*}
and if the (co)magma have (co)units then (assuming $S\neq \varnothing$ in $\fset$ or $\frel$), 
\begin{equation*}
    \tikzfigscale{1}{figs/putgeta3}
\end{equation*}
\end{enumerate}
\end{proof}

Thus an update satisfying PutGet A would imply a very exotic magma-comagma: they cannot have the GetPut rule in finite dimensions and even dropping this rule, they cannot have a (co)unit.

\begin{prop}\label{prop:putgetB}
    PutGet B makes the comagma coassociative under the Put $\putt{}$:
    \begin{equation*}
    \tikzfigscale{1}{figs/coassput}
    \end{equation*}
\end{prop}
\begin{proof}
Starting with the left-hand side and applying PutGet B twice followed by GetGet and PutGet B in reverse recovers the right-hand side.
\end{proof}

\begin{prop}\label{prop:putgetC}
    PutGet C makes the magma associative under the Get $\get{}$:
    \begin{equation*}
    \tikzfigscale{1}{figs/assput}
    \end{equation*}
\end{prop}
\begin{proof}
Starting with the left-hand side and applying PutGet C twice followed PutPut and PutGet C in reverse recovers the right-hand side.
\end{proof}

\begin{prop}\label{prop:putgetBC}
    Demanding PutGet B and C makes the magma and comagma Frobenius under the Put (and the Get):
    \begin{equation*}
    \tikzfigscale{1}{figs/frobput}
    \end{equation*}
\end{prop}
\begin{proof}
Consider $\get{}\circ \putt{} \circ (\putt{}\otimes \id{p})$ and apply the following laws:
\begin{itemize}
\item PutGet B followed by PutPut, for the left-hand side
\item PutPut followed by PutGet B, for the centre
\item PutGet C followed by PutGet B, for the right-hand side
\end{itemize}
\end{proof}

Furthermore, there are conditions the magma and comagma inherit via reference to the (co)module structure alone.

\begin{defn}[Commutative] Put (Get) is commutative if
\begin{equation}\label{commuteget}
    \tikzfigscale{1}{figs/commuteget}
\end{equation}
\end{defn}
\begin{prop}\label{prop:commute}
    Commutativity of Put (Get), makes the magma (comagma) commutative under the Put (Get):
    \begin{equation*}
    \tikzfigscale{1}{figs/propcommute}
    \end{equation*}
\end{prop}
\begin{proof}
By use of PutPut on both sides of the equation.
\end{proof}

\begin{prop}\label{prop:ignorable}
    A trivial update (outcome) gives the magma (comagma) a unit (counit) under the Put (Get):
    \begin{equation*}
    \tikzfigscale{1}{figs/propunit}
    \end{equation*}
\end{prop}
\begin{proof}
By use of PutPut in the 2nd and 3rd expressions.
\end{proof}

\section{Quantum Measurements As Weak Update Structures}
We now show that decoherence preserves the equations of a weak update structure, and that the projector valued spectra used to define quantum measurements also define strong update structures. As a corollary, every quantum measurement can therefore be used to define an update structure. 
\subsection{Applying Decoherence to Update Structures}
That applying decoherence to the property wire produces a new update structure will be a corollary of the following general result. 
\begin{prop}[Transformation on Weak Update structures]\label{transupdate}
Given an update structure $(\putt{},\get{},\mix{},\copyy{})$
and a morphism $m:p\morph{}p$ satisfying, 
\begin{equation}
    \tikzfigscale{1}{figs/magmahom2}
\end{equation}
one can define a weak update structure by the following:
\begin{equation}
    \tikzfigscale{1}{figs/magmaupdate2}
\end{equation}
\end{prop}
\begin{proof}
We give the proof that the Repeat Update axiom is preserved
\begin{equation}
    \tikzfigscale{1}{figs/magmarepeatupdate2}
\end{equation}
The rest of the axioms follow similarly.
\end{proof}
As an important special case the above proposition applies to any $m:p \rightarrow p$ satisfying:
\begin{equation}
    \tikzfigscale{1}{figs/magmahom}
\end{equation}
that is, any idempotent magma co-magma homomorphism. For any $\dagger$-compact category $\cat{C}$ the functor $F:\cat{C}\morph{}\cpm{\cat{C}}$ which takes $f: A \rightarrow B$ to $f \otimes f^{*}: A \otimes A^{*} \rightarrow B \otimes B^{*}$, maps any update structure in $\cat{C}$ to an update structure in $\cpm{\cat{C}}$. Furthermore, for any special commutative $\dag$-Frobenius algebra $\frob{}$, the decoherence morphism $\deco{}$ is an idempotent magma homomorphism for $F(\frob{})$.
\begin{corr} \label{prop:decoupdate}
Given any weak update structure $(\putt{},\get{},\frob{},\copyy{})$ in a $\dagger$-compact category $\cat{C}$ on system $S$, and property $p$ with $\frob{}$ a special commutative $\dag$-Frobenius algebra, composition with the decoherence map after the functor $F: \cat{C} \rightarrow \cpm{\cat{C}}$,
\begin{equation}\label{decoupdate}
    \tikzfigscale{1}{figs/decohereupdate}
\end{equation}
generates a weak update structure on $F(S) = S\otimes S^*$ and $F(p)=p\otimes p^*$.
\end{corr}
Furthermore, update structures obtained by decoherence define update structures on formally classical objects in $\spl{\cpm{\cat{C}}}$ - the Karoubi envelope of $\cpm{\cat{C}}$. 
\begin{defn}[Karoubi Envelope]
The Karoubi envelope $\spl{\cat{C}}$ of a category $\cat{C}$ has as objects the pairs $(A,\pi)$ where $A \in \ob{\cat{C}}$ and $\pi: A \rightarrow A$ is an idempotent. The morphisms $f : (A ,\pi) \rightarrow (B,\sigma)$ are the morphisms $f:A \rightarrow B$ such that $\sigma \circ f = f = f \circ \pi$. 
\end{defn}
As studied in \cite{coecke_classicality, selinger_idempotents, Heunen_cp}, the object $(F(p),\deco{})$ represents a classical version of the object $p$, on which the only states and transformations permitted are those which are unaffected by decoherence $f \circ \deco{} = f = \deco{} \circ f$. The weak update structure defined in equation (\ref{decoupdate}) then also appears as a weak update structure in $\spl{\cpm{\cat{C}}}$ on system $(F(S),\id{F(S)})$ and property $(F(p),\deco{})$.

\subsection{Quantum Observables Define Update Structures}
Projector valued spectra in $\fhilb$ capture quantum observables, they are characterised by the following graphical conditions.

\begin{defn}[Projector-Valued Spectrum] \cite{coecke_measurements}
The pair of morphisms $(\Pi,\wfrob{})$ in $\fhilb$ is a projector-valued spectrum if $\wfrob{}$ is a special commutative $\dag$-Frobenius algebra and $\Pi$ satisfies: 
\begin{equation}
    \tikzfigscale{1}{figs/pspectrum}
\end{equation}
These equations are paraphrased by asking that $\Pi$ be $p$-idempotent, $p$-self-adjoint, and $p$-complete respectively. Projectors are recovered from $(\Pi,\wfrob{})$ by inserting states in the basis $\{\ket{i}\}$ associated to $\wcopyy{}$.

\begin{equation}
    \tikzfigscale{1}{figs/projector}
\end{equation}

Since $\wcopyy{}$ copies the elements of $\{\ket{i}\}$, $p$-idempotency implies orthogonality and idempotency of the $P_{i}$, $p$-self-adjointness implies adjointness of each $P_{i}$ and $p$-completeness implies that the projectors sum to the identity.
\end{defn}
A projector valued spectrum defines a quantum observable, when the observable wire reads $\ket{i}$, the system will be found in the image of $P_{i}$. Using the $\dagger$ we can construct an update structure from a projector valued spectrum.
\begin{prop} For every projector-valued spectrum $(\Pi,\wfrob{})$ , the tuple $(\Pi,\Pi^{\dagger},\wfrob{},\wcopyy{})$ is an update structure.
\end{prop}
\begin{proof}
By $p$-idempotency, $(\Pi,\Pi^{\dagger},\wfrob{},\wcopyy{})$ is immediately a (co)magma module pair. The GetPut axiom holds since any projector valued spectrum is isometric with inverse given by the adjoint:
\begin{equation}
    \tikzfigscale{1}{figs/getputprojector}
\end{equation}
Finally the PutGet axiom holds.
\begin{equation}
    \tikzfigscale{1}{figs/putgetprojector2}
\end{equation}
\end{proof}
Quantum measurements $(Q(\Pi{}),Q(\wfrob{}))$ are defined by taking projector valued spectra $(\Pi{},\wfrob{})$ in $\fhilb$, embedding into $\cpm{\fhilb}$, $F \circ \dagger: \fhilb \rightarrow \cpm{\fhilb}$ \cite{selinger} and subsequently applying the decoherence morphism $\deco{}$ to the property wires \cite{coecke_measurements, bob_coecke_aleks_kissinger_picturing_2017}:
\begin{equation*}
    \tikzfigscale{1}{figs/qmmeasure3}
\end{equation*}


\begin{corr}
For every quantum measurement $(Q(\Pi) ,Q(\wfrob{}))$ the tuple $(Q(\Pi)^{\dagger}, Q(\Pi) ,Q(\wfrob{})^{\dagger},Q(\wfrob{}))$ is a weak update structure.
\end{corr}

The GetPut axiom fails for quantum measurements because extraction and re-insertion of classical information induces decoherence, a disturbance, on the quantum system. Via a GetPut restriction as introduced in \cite{hefford2020categories} any such weak update structure on a fully quantum system $(F(S),\id{F(S)})$ can be used to construct a strong update structure on a measured, partially decohered, system $(S,\putt{} \circ \get{})$. On top of being update structures, projector valued spectra carry some additional conditions, they have trivial updates (outcomes) their Puts are commutative, and they are faithful.  
\begin{defn}[Faithful]
An update structure is faithful if 
\begin{equation}\label{faithful}
    \tikzfigscale{1}{figs/faithful}
\end{equation}
\end{defn}
Conversely, the above conditions are enough to characterise projector valued spectra.

\begin{prop} \label{putisproj}
 An update structure of the form $(\putt{}, \putt{}^{\dagger},\frob{}, \copyy{})$ in $\cat{FHilb}$ is a projector valued spectrum if and only if it is faithful, commutative, and has a trivial update (outcome).
\end{prop}
\begin{proof}
Since we have already demonstrated the $\Rightarrow$ direction we only need to consider $\Leftarrow$. The results of section 3 imply that the magma of a faithful, commutative $\dagger$-update with a trivial update (outcome) is a commutative $\dag$-Frobenius algebra, furthermore by applying the PutPut axiom to the left hand side of the repeat update axiom, this Frobenius algebra is special. The unit $\unit{}$ is the trivial update, so $\putt{}$ is $p$-complete. The PutPut axiom is precisely the demand of $p$-idempotency, and finally since $\putt{}$ is $p$-complete, the PutGet axiom implies that $\putt{}$ is $p$-self-adjoint. 
\begin{equation}
    \tikzfigscale{1}{figs/weakisselfadjoint}
\end{equation}
\end{proof}
To search for new update structures in $\cat{FHilb}$ we must then relax one of the conditions of proposition \ref{putisproj}.

\section{Examples of Update Structures}
To demonstrate the kinds of procedures that can be implemented as update structures we give some new examples in $\fhilb$, $\cpm{\fhilb}$, and $\set$. 
\subsection{Morphisms as Properties}
We now give a new example of a non-commutative strong update structure in $\fhilb$ based on the ``pair of pants" monoid \cite{heunen2019categories}. $\dagger$-Compact categories come equipped with an evaluation morphism $\epsilon : S^{*} \otimes S \otimes S \rightarrow S$ and a composition morphism $\sigma : S^{*} \otimes S \otimes S^{*} \otimes S \rightarrow S^{*} \otimes S$ which satisfy PutPut. The evaluation takes as inputs a morphism and a state, and produces an output by applying the morphism to the state. We can use $\epsilon$ and $\sigma$ as $\putt{}$ and $\mix{}$ respectively, along with the dagger to construct an update structure for which systems are states, and properties are morphisms waiting to be applied to those states. 
\begin{prop}
The tuple $(\putt{},\get{},\mix{},\copyy{})$ in a $\dagger$-compact category with invertible scalars defined by,
\begin{equation*}
    \tikzfigscale{1}{figs/compupdate}
\end{equation*}
is a strong update structure.
\end{prop}
\begin{proof}

PutPut and GetGet can be confirmed graphically, but also follow by definition of composition $\sigma$ as the adjunct to $\epsilon \circ (\epsilon \otimes I)$ in any symmetric closed monoidal category. PutGet is quick to verify 
\begin{equation}\label{pgcomp}
    \tikzfigscale{1}{figs/putgetcomp}
\end{equation}
as is GetPut
\begin{equation}\label{gpcomp}
    \tikzfigscale{1}{figs/getputcomp}
\end{equation}

\end{proof}
This update structure is faithful and has a trivial update, its magma and co-magma are Frobenius, but non-commutative. We note that the $\putt{}$ of this update structure was introduced in \cite{coecke_mathematical_meaning} as a way to make ``fuzz" and ``phaser" updates internal.

\subsection{A Security Tagged Database}
Update structures generalise vwb lenses by allowing for a Get of type $\get{}:S \rightarrow S \otimes p$ which returns the system $S$ in addition to the property $p$. This allows us to create lens like structures for which the act of tampering has an influence on the system. We present a toy example in $\set$ by introducing a security feature which raises a flag to indicate that a database has been breached. We imagine that the system decomposes as a database entry and a flag $S = p \otimes F$ with $F = \{\texttt{safe},\texttt{breached}\}$. The flag records whether or not anybody has breached the database since it was last set to $\texttt{safe}$. In particular we take $\putt{}:  p \otimes F \otimes p \rightarrow p \otimes F$ and $\get{}: p \otimes F \rightarrow p \otimes F \otimes p$ defined by:

\begin{align*}
    \putt{} & :: \langle p,x \rangle \times p' \mapsto \langle p',\texttt{breached} \rangle & \mix{} & :: p \times p' \mapsto \langle p' \rangle \\
    \get{} & :: \langle w,x\rangle \mapsto \langle p,\texttt{breached} \rangle \times p & \copyy{} & :: \langle p \rangle \mapsto p \times p
\end{align*}

It is easy to check that this choice defines a weak update structure, the GetPut axiom fails precisely because the act of interacting with the database has influence on the database, beyond the property $p$. In a restricted case for which there is only a flag for the occurrence of an update $S = p \otimes F_{\putt{}}$, this update structure is a lens, although not a vwb lens. Such a weak update structure can be used to build a new strong update structure in $\spl{\set}$ on the $\texttt{breached}$ type $(S,\putt{} \circ \get{})$ \cite{hefford2020categories}.

\subsection{Quantum Versions of Database Updates}
We now consider update structures in $\fhilb$ or $\cpm{\fhilb}$ with $(\get{},\copyy{})$ in the form of a projector valued spectrum or quantum measurement, and with a vwb flavoured ignore-replace co-magma $\mix{} = \effect{} \otimes f$. In $\fhilb$ a natural choice is for $\get{}$ to copy a subsystem $S_2$ of $S$ in a particular basis, and for $\putt{}$ to delete in the same basis.
\begin{equation}\label{qmdatabaseproj}
    \tikzfigscale{1}{figs/qmdatabasefhilb}
\end{equation}
where $\wfrob{}$ is a special commutative $\dag$-Frobenius algebra. We note that $\putt{}$ and $\mix{}$ both require post-selection on the unit for $\wfrob{}$. The embedding $F: \fhilb \rightarrow \cpm{\fhilb}$ followed by application of decoherence to every property wire as in proposition \ref{transupdate} generates a new update structure, to implement this new $\putt{}$ would still require a post-selection on $F(S_2)$.

We can directly define a vwb lens-like update structure with a $\putt{}$ which deterministically discards \cite{bob_coecke_aleks_kissinger_picturing_2017} and replaces the property $F(S_2)$ and with $(\get{},\copyy{})$ a quantum measurement.



\begin{equation}\label{qmdatabase}
    \tikzfigscale{1}{figs/qmdatabasecausal}
\end{equation}
The axioms are easily checked directly, they also follow by noting that this update structure can be constructed from our projector valued spectrum by applying decoherence to each property wire followed by applying GetPut to each system wire \cite{hefford2020categories}, each of which preserves the axioms of a weak update structure. This update structure is weak because the system $F(S_1) \otimes F(S_2)$ is partially decohered when information about $F(S_2)$ is extracted and re-inserted. The reduced process on the system obtained by tracing out the property wire of $\get{}$ is a partial decoherence, whereas for any $\get{}$ built from composing a deterministic $g: F(S_1) \otimes F(S_2) \rightarrow F(S_2)$ and a copy map (i.e. in a lens-like way), the reduced process on the system
\begin{equation}\label{qmdatabase}
    \tikzfigscale{1}{figs/detget2}
\end{equation}
is a full decoherence on the entire system. 

\color{black}


\section{Conclusion}
In this article we have furthered our development of update structures. We originally introduced update structures with operational inspiration and we have now shown that very-well-behaved lenses, quantum observables and measurements and a slight alteration on the pair-of-pants monoid all live within the zoo of update structures. This offers a unifying perspective of two fields and elevates update structures to algebraic objects that deserve further investigation.

This work opens several lines of investigation. We have not discussed composition of update structures in this article, but there seems to be multiple interesting such notions. Each such notion would impose additional structure in order to be well-defined. Exploring these alternatives would provide further insight into those structures and contribute to a more featureful unification.

We also endeavour to expand the zoo of update structures with more exotic beasts. Perhaps such structures could allow for updates which behave neither quantumly nor classically and further the search for post-quantum phenomena.

A further goal would be to formally justify the axioms of update structures. Currently they are operationally well-chosen but it is not clear how necessary and sufficient they are, or whether other choices could lead to related and interesting structures. We touched upon this when we investigated alternative PutGet rules, but the full ramifications of these choices are not currently understood.

\subsubsection*{Acknowledgements}
With many thanks to Sean Tull for pointing out connections to quantum observables and to Bob Coecke for helpful discussions.
MW and JH acknowledge the support of University College London and the EPSRC [grant number EP/L015242/1].

\bibliographystyle{eptcs}
\bibliography{bibliography}

\end{document}